\documentclass{rQUF2e} 

\usepackage{epstopdf}
\usepackage{subfigure}

\theoremstyle{plain}
\newtheorem{theorem}{Theorem}[section]

\newtheorem{proposition}[theorem]{Proposition}

\theoremstyle{definition}

\theoremstyle{remark}
\newtheorem{remark}{Remark}

\usepackage{tikz}
\usetikzlibrary{arrows,matrix,positioning,patterns}
\usetikzlibrary{calc}
\usetikzlibrary{plotmarks}

\usepackage[utf8]{inputenc}
\usepackage{pgfplots}
\usepgfplotslibrary{groupplots}

\usepackage{mathtools}
\usepackage{multirow} 

\usepackage{algorithm} 
\usepackage{algpseudocode}
\newcommand\NoDo{\renewcommand\algorithmicdo{}}
\newcommand\NoThen{\renewcommand\algorithmicthen{}}

\usepackage{natbib}
\setcitestyle{square}
\renewcommand{\refname}{Bibliography}

\begin{document}


\title{A posteriori multi-stage optimal trading under transaction costs and a diversification constraint}

\author{Mogens Graf Plessen$^\ast\dag$\thanks{$^\ast$Corresponding author.
Email: mogens.plessen@imtlucca.it} and Alberto Bemporad${\dag}$
\affil{$\dag$IMT School for Advanced Studies Lucca, Piazza S. Francesco 19, 55100 Lucca, Italy}}
 
\maketitle

\begin{abstract}
This paper presents a simple method for a posteriori (historical) multi-variate multi-stage optimal trading under transaction costs and a diversification constraint. Starting from a given amount of money in some currency, we analyze the stage-wise optimal allocation over a time horizon with potential investments in multiple currencies and various assets. Three variants are discussed, including  unconstrained trading frequency, a fixed number of total admissable trades, and the waiting of a specific time-period after every executed trade until the next trade. The developed methods are based on efficient graph generation and consequent graph search, and are evaluated quantitatively on real-world data. The fundamental motivation of this work is preparatory labeling of financial time-series data for supervised machine learning.    
\end{abstract}



\section{Introduction}

Algorithmic assistance to traders and portfolio managers has become standard practice. It can be distinguished between \emph{algorithmic trading}, i.e., fully automated high- or low-frequency trading, and \emph{algorithmic screening} or semi-automated high- or low-frequency trading with computer programs providing recommendations to the human trader. Both algorithmic trading and screening are fundamentally based on predictions of future developments. Predictions may be made based on, for example, financial accountancy, technical chart analysis, global macroeconomic analysis, news, sentiments and combinations thereof. There exists a plethora of literature on \emph{financial times-series forecasting}. For methods based on support vector machines, see, for example, \cite{tay2001application}, \cite{kim2003financial}, \cite{van2001financial}, and \cite{chowdhury2018short}. In general, influential factors on trading decisions are trading frequency, targeted time horizons, performance expectations, asset choices, foreign exchange rates, transaction costs and risk-management, for example, in the form of investment diversification. Our paper belongs to the class of technical chart analysis. The data on which the analysis is based are daily adjusted  closing-prices of various currencies and assets.  Short-selling, borrowing of money, and the trading of derivatives are not treated, eventhough the presented methodologies can be extended to include them.

The motivation and contribution of this paper is threefold: i) the development of a simple algorithm for a posteriori (historical) multi-variate multi-stage optimal trading under transaction costs and a diversification constraint, including the discussion of unconstrained trading frequency, a fixed number of total admissable trades, and the waiting of a specific time-period after every executed trade until the next trade; ii) the quantification of the effects of transaction costs on a posteriori optimal trading evaluated on real-world data; and finally iii) the preparatory labeling of financial time-series data for supervised machine learning.

This paper is related most closely to the work of \cite{boyarshinov2010efficient}, who discuss a dynamic programming solution to the optimal investment in either one stock or one bond under consideration of unconstrained trading frequency, and a bound on the admissable number of trades. Additionally, a method for optimization of Sterling and Sharpe ratio are presented. No real-world data analysis is conducted though. Additional differences are our discussion of a diversification constraint, the constraint of introducing a waiting period after every executed trade until the next trade, and a synchronous trading constraint. Furthermore, we introduce a heuristic for each of the constrained optimal investment problems (with a bound on the admissable number of trades and a waiting period constraint), thereby reducing the computational complexity of the methods while not compromising optimality of the resulting solution. For an overview of measures to reduce risk by the introduction of various constraints, for example, on the drawdown probability or shortselling, see \cite{lobo2007portfolio} where \emph{one-step ahead} optimization is conducted, importantly, based on estimates of one-step ahead returns and covariance matrices of a set of risky assets. In contrast this papers is concerned about \emph{multi-stage} optimization and historical optimal trading with hindsight. The mathematical approaches therefore differ significantly (convex optimization vs. graph search). Optimal trading based on stochastic models, usually stochastic differential equations (SDEs), and the consideration of fixed and/or proportional transaction costs is treated, for example, in \cite{altarovici2015asymptotics}, \cite{lo2001asset}, \cite{morton1995optimal} and \cite{korn1998portfolio}. In contrast, this paper is data-based only, i.e., without consideration of any mathematical model explaining the generation of this data. For a discussion about the existence of trends in financial time series, see \cite{fliess2009mathematical}. For the general discussion of \emph{transcation cost analysis} (TCA), see \cite{gomes2010transaction}, and further \cite{kissell2013science}, \cite{kissell2008transaction} and \cite{kissell2006expanded} for a discussion about how TCA can be used by portfolio managers to improve performance and the development of a framework for pre-, intra- and post-trade analysis.

The remainder of this paper is organized as follows. Section \ref{sec_sys_mdling} discusses transition dynamics modeling and introduces notation. Multi-stage optimization without a diversification constraint is treated in Section \ref{sec_multi_stage_opt_without_diversification}, whereas Section \ref{sec_multi_stage_opt_with_diversification} includes a diversification constraint. Section \ref{sec_NumEx} presents numerical examples based on real-world data, before concluding with Section \ref{sec_Conclusion}.

\section{One-stage modeling of transition dynamics \label{sec_sys_mdling}}

\subsection{Notation\label{subsec_states_notation}}

Let time index $t\in\mathbb{Z}_+$ be associated with the trading period $T_s$, such that trading instants are described as $tT_s$, whereby $T_s$ may typically be, for example, one week, one day, or less (for intraday trading). The system state $\bm{z_t}$ at time $t$ is defined as an eight-dimensional vector of mixed integer and real-valued quantities,    
\begin{equation}
\bm{z_t} = \begin{bmatrix} i_t & k_t & j_t & m_t^{c_t} & n_t & w_t^0 & d_t & c_t \end{bmatrix},\label{eq_def_z_t}
\end{equation}
where $i_t\in \mathcal{I} = (\mathcal{I}_{N_c} \cup \mathcal{I}_{N_a})$ denotes investment identification numbers partitioned into  $N_c$ currencies and $N_a$ different risky non-currency assets, such that $\mathcal{I}_{N_c}=\{0,1,\dots,N_c-1\}$ and $\mathcal{I}_{N_a}=\{N_c,\dots,N_c+N_a-1\}$. For ease of reference, in the following, we lump currencies and non-currency assets in the term \emph{asset} and only distinguish when context-necessary. The integer number of conducted trades along an \emph{investment trajectory} shall be denoted by $k_t\in\mathbb{Z}$, whereby an investment trajectory is here defined as a sequence of states $z_t$, $t=0,1,\dots,N_t$, where $N_t$ is the time horizon length. Let $j_t$ denote the investment identification number preceding $i_t$ at time $t-1$ (parent node), i.e., $j_t=i_{t-1}$. We define $m_t^{c_t}\in\mathbb{R}_+$ as the real-valued and positive \emph{cash position} (liquidity) held in the currency identified by $c_t\in\mathcal{I}_{N_c}$. The number $n_t\in\mathbb{Z}_+$ indicates the number of non-currency assets held. The current wealth, composed of cash position and non-currency asset, is denoted by $w_t^0$ and shall always be in monetary units EUR. Euro is considered as our reference currency and shall throughout this paper be identified by $i_t=0$. The integer number of time samples since the last trade is defined by $d_t\in\mathbb{Z}_+$. The (unitless) foreign exchange (fx) rate $x_t^{c_1,c_2}$ for two currencies $c_1\in\mathcal{I}_{N_c}$ and $c_2\in\mathcal{I}_{N_c}$ is defined as $x_t^{c_1,c_2}$ such that $m_t^{c_2} = m_t^{c_1} x_t^{c_1,c_2}$.
Thus, $m_t^{c_1}$ and $m_t^{c_2}$ have numerical values, however, with units identified by $c_1\in\mathcal{I}_{N_c}$ and $c_2\in\mathcal{I}_{N_c}$, respectively. Non-currency asset prices are denoted by $p_t^{c_t,a}$, whereby $c_t$ identifies the price unit and $a\in\mathcal{I}_{N_a}$ the asset. We treat foreign exchange rates and asset prices as time-varying parameters obtained from data. In the sequel, various sets of admissable system states are defined. For brevity, we therefore use a shorthand notation. For example, we define a set as $\mathcal{Z}_t=\left\{ \bm{z_t}: i_t=10 \right\}$, implying $\mathcal{Z}_t=\left\{ \bm{z_t}: i_t=10,~\text{and $i_t$ associated with $\bm{z_t}$ according \eqref{eq_def_z_t}}\right\}$. 

\subsection{Transaction costs}

For the modeling of transaction costs, we follow the notion of \cite{lobo2007portfolio}, modeling transaction costs as non-convex with a fixed charge for any nonzero trade (fixed transaction costs) and a linear term scaling with the quantity traded (proportional transaction costs). Thus, for a foreign exchange at time $t-1$, we model $
m_t^{c_t} = m_{t-1}^{c_{t-1}}x_{t-1}^{c_{t-1},c_t}(1-\epsilon_\text{fx}^{c_{t-1},c_t}) - \beta_\text{fx}^{c_{t-1},c_t}$, where $\epsilon_\text{fx}^{c_{t-1},c_t}$ and $\beta_\text{fx}^{c_{t-1},c_t}$ denote the linear term and the fixed charge, respectively. Similarly, transaction costs for transactions from currency to non-currency asset, between assets of different currencies and the like can be defined. We can further differentiate between linear terms for buying and selling. To fully introduce notation for transaction costs ($\epsilon_\text{buy}^{i_t},\beta_\text{buy}^{i_t}\geq 0$), we state the transaction from a cash position towards an asset investment and vice versa. For a transaction of buying $n_{t-1}$ of asset $i_{t-1}$ at time $t-1$, we obtain
\[
m_t^{c_t} = m_{t-1}^{c_{t-1}}x_{t-1}^{c_{t-1},c_t}(1-\epsilon_\text{fx}^{c_{t-1},c_t}) - \beta_\text{fx}^{c_{t-1},c_t} - n_{t-1}p_{t-1}^{c_t,i_t}(1+\epsilon_\text{buy}^{i_t}) - \beta_\text{buy}^{i_t}.
\]
For a transaction of selling $n_{t-1}$ of asset $i_{t-1}$ and transforming to currency $c_t$, we obtain
\[
m_t^{c_t} = \left( m_{t-1}^{c_{t-1}} + n_{t-1}p_{t-1}^{c_{t-1},i_{t-1}}(1-\epsilon_\text{sell}^{i_{t-1}}) - \beta_\text{sell}^{i_{t-1}} \right)x_{t-1}^{c_{t-1},c_t}(1-\epsilon_\text{fx}^{c_{t-1},c_t}) - \beta_\text{fx}^{c_{t-1},c_t}.
\]

Finally, note that transaction costs may vary dependent on the assets involved. 
\subsection{Transition dynamics\label{subsec_transition_dynamics}}

Given our assumption of being able to invest in currencies and non-currency assets, there are six general types of transitions dependent on the investment at time $t-1$. For an introduction to \emph{Markov Decision Processes} (MDP), see \cite{puterman2014markov}. We initialize $\bm{z_0} = \begin{bmatrix} 0 & 0 & 0 & m_0^{0} & 0 & m_0^0 & 0 & 0 \end{bmatrix}$. Then, the transition dynamics are 
\begin{equation}
\bm{z_t} = \begin{cases} \bm{ z_t^{(1)}},~\text{if}~\{i_t:i_t=i_{t-1},~z_{t-1}~\text{with}~i_{t-1}
\in\mathcal{I}_{N_c}\},\\
\bm{z_t^{(2)}},~\text{if}~\{i_t:i_{t}\in\mathcal{I}_{N_c}\backslash \{i_{t-1}\},~z_{t-1}~\text{with}~i_{t-1}\in\mathcal{I}_{N_c} \},\\
\bm{z_t^{(3)}},~\text{if}~\{i_t:i_{t}\in\mathcal{I}_{N_a},~z_{t-1}~\text{with}~
i_{t-1}\in\mathcal{I}_{N_c} \},\\
\bm{z_t^{(4)}},~\text{if}~\{i_t:i_t=i_{t-1},~z_{t-1}~\text{with}~
i_{t-1}\in\mathcal{I}_{N_a}\},\\
\bm{z_t^{(5)}},~\text{if}~\{i_t:i_{t}\in\mathcal{I}_{N_c},~z_{t-1}~\text{with}~
i_{t-1}\in\mathcal{I}_{N_a} \},\\
\bm{z_t^{(6)}},~\text{if}~\{i_t:i_{t}\in\mathcal{I}_{N_a}\backslash \{i_{t-1}\},~z_{t-1}~\text{with}~
i_{t-1}\in\mathcal{I}_{N_a} \},
 \end{cases}\label{eq_def_zt}
\end{equation}
where $z_t^{(j)},~\forall j=1,\dots,6$, is defined next, and our control variable $u_{t-1}$ is the targeted investment identified by variable $i_t$, i.e., $u_{t-1}=i_t$. We have
\begin{align*}
\bm{z_t^{(1)}} &= \begin{bmatrix} i_{t-1} & k_{t-1} & j_{t-1} & m_{t-1}^{c_{t-1}} & 0 & w_{t-1}^0 & \tilde{d}_t & c_{t-1}\end{bmatrix},\\ 
\bm{z_t^{(2)}} &= \begin{bmatrix} i_t & k_{t-1}+1 & j_{t-1} &  \varphi(m_t^{c_t}) & 0 & m_t^{c_t}x_t^{c_t,0} & \tilde{d}_t & i_t\end{bmatrix},\\
\bm{z_t^{(3)}} &= \begin{bmatrix} i_t & k_{t-1}+1 & j_{t-1} & \tilde{m}_{t}^{c_{t}} & \tilde{n}_t& \tilde{w}_t^0 & \tilde{d}_t & c(i_t) \end{bmatrix},
\end{align*}
with $c(i_t)$ denoting the currency of asset $i_t$ and with
\begin{align*}
\tilde{d}_t &= \begin{cases} d_{t-1}+1, & \text{if}~d_{t-1}<D-1,\\
0 & \text{otherwise,} \end{cases}\\
\varphi(m_t^{c_t}) &= m_{t-1}^{c_{t-1}}x_{t-1}^{c_{t-1},c_t}(1-\epsilon_\text{fx}^{c_{t-1},c_t}) - \beta_\text{fx}^{c_{t-1},c_t},
\end{align*}
and where variable $D$ determines an overflow in $d_t$ and will become relevant when later discussing the constraint of waiting a specific amount of time until the next admissable trade. Furthermore, $\tilde{m}_t^{c_t}$ and $\tilde{n}_t$ are obtained from solving 
\begin{equation}
\max_{m_t^{c_t}\geq 0} ~ \left\{ n_t: n_{t}= \frac{m_{t-1}^{c_{t-1}}x_{t-1}^{c_{t-1},c_t}(1-\epsilon_\text{fx}^{c_{t-1},c_t}) - \beta_\text{fx}^{c_{t-1},c_t} - \beta_\text{buy}^{i_t} - m_t^{c_t} }{ p_{t-1}^{c_t,i_t}(1+\epsilon_\text{buy}^{i_t})},~ n_{t} \in\mathbb{Z}_+  \right\},\label{eq:OP_zt(3)_mntilde}
\end{equation} 
with $\tilde{m}_t^{c_t}$ denoting the optimizer and $\tilde{n}_t$ the corresponding optimal objective function value. Thus, given $m_{t-1}^{c_{t-1}}$, we find the largest possible positive integer number of assets we can purchase under the consideration of transaction costs. The (small) cash residual is then $\tilde{m}_t^{c_t}\geq 0$. Therefore, for the portfolio wealth at time $t$ in currency EUR, we obtain $\tilde{w}_t^0 = (\tilde{m}_t^{c_t} + \tilde{n}_t p_t^{c_t,i_t})x_t^{c_t,0}$. Furthermore,  
\begin{align*}
\bm{z_t^{(4)}} &= \begin{bmatrix} i_{t-1} & k_{t-1} & j_{t-1} & m_{t-1}^{c_{t-1}} & n_{t-1}  &w_{t-1}^0 &  \tilde{d}_t& c_{t-1} & \end{bmatrix},\\
\bm{z_t^{(5)}} &= \begin{bmatrix} i_t& k_{t-1}+1 & j_{t-1}  & \phi(m_t^{c_t}) & 0 & m_t^{c_t}x_t^{c_t,0} & \tilde{d}_t & i_t & \end{bmatrix},\\
\bm{z_t^{(6)}} &= \begin{bmatrix} i_t&  k_{t-1}+1 & j_{t-1} & \bar{m}_{t}^{c_{t}} & \bar{n}_t & \bar{w}_t^0 & \tilde{d}_t & c(i_t) \end{bmatrix},
\end{align*}
with 
\begin{equation*}
\phi(m_t^{c_t}) = \left( m_{t-1}^{c_{t-1}} + n_{t-1}p_{t-1}^{c_{t-1},i_{t-1}}(1-\epsilon_\text{sell}^{i_{t-1}}) - \beta_\text{sell}^{i_{t-1}} \right)x_{t-1}^{c_{t-1},c_t}(1-\epsilon_\text{fx}^{c_{t-1},c_t}) - \beta_\text{fx}^{c_{t-1},c_t},
\end{equation*}
and where $\bar{m}_t^{c_t}$ and $\bar{n}_t$ are obtained from solving 
\begin{equation}
\max_{m_t^{c_t}\geq 0}~\left\{ n_t: n_{t}= \frac{ \phi(m_t^{c_t}) - \beta_\text{buy}^{i_t} - m_t^{c_t}  }{ p_{t-1}^{c_t,i_t}(1+\epsilon_\text{buy}^{i_t})},~n_{t} \in\mathbb{Z}_+ \right\},
\label{eq:OP_zt(6)_mnbar}
\end{equation} 
with $\bar{m}_t^{c_t}$ denoting the optimizer and $\bar{n}_t$ the corresponding optimal objective function value. Then, for the portfolio wealth at time $t$ in EUR, we obtain $\bar{w}_t^0 = (\bar{m}_t^{c_t} + \bar{n}_t p_t^{c_t,i_t})x_t^{c_t,0}$.
 
The solution to \eqref{eq:OP_zt(3)_mntilde} and \eqref{eq:OP_zt(6)_mnbar} can be easily computed by setting $m_t^{c_t}$ initially zero, then rounding the corresponding real-valued $n_t$ to the largest smaller integer, before then computing the cash residuals, respectively. The methodology of preserving a cash residual is implemented in order to enforce an integer-valued number of shares in assets. 

\subsection{Remarks about optimality and transition dynamics modeling}

An investment trajectory is defined as a sequence of states $z_t$, $t=0,1,\dots,N_t$. We wish to find an optimal (in the sense of wealth-maximizing) investment trajectory. Several remarks about above problem formulation and transition dynamics modeling can be made. 

First, suppose all of the initial money $m_0^0$ is fully allocated to the optimal investment trajectory, then there is no diversification present, and, defining the final return as $r_{N_t} = (w_{N_t}^0 - m_0^0)/m_0^0$, the optimal investment trajectory never returns less than $r_{N_t}<0\%$. This is since one feasible investment trajectory is to remain invested in the initial reference currency (EUR) for all $t=0,1,\dots,N_t$. This can be taken into account as a heuristic for transition graph generation.

Second, above transition dynamics modeling naturally results in cash residuals when investing in non-currency assets. According to our modeling, the cash residuals are enforced to be in the currency of the purchased asset. This may be suboptimal when the non-currency asset in which we invest is extremely expensive (e.g., worth thousands of EUR), since the resulting cash-residuals may then be very large. Then, in general it may be worthwhile to invest the cash residual into another asset which is more profitable than the ``enforced'' residual currency. Two comments are made. On the one hand, assets with such prices are rare in practice. On the other hand, and more importantly, in order to admit free investing of cash residuals, an extension of the state space (beyond 8 variables) would be required such that any cash residual could be invested in any of the $N_c+N_a-1$ assets. Then, $N_c+N_a-1$ additional branches would need to be added to the transition graph, which, in the most general case, would also need further branching at subsequent stages. This  considerably complicates the tracking of states and is therefore not applied in the following.   

Third, transition dynamics \eqref{eq_def_zt} indicate an \emph{all-or-nothing} strategy. At every time $t$, the investment at that time is maintained or, alternatively, reallocated to exactly \emph{one}--the most profitable--currency or asset, whereby cash residuals are accounted for as described in the previous paragraph.

Fourth, let us briefly discuss the effect of absence of transaction costs on optimal trading frequency. For simplicity let us consider the case of being able to invest in an asset of variable value (such as a stock) and holding of cash in the currency in which the risky asset is traded. Relevant discrete-time dynamics can then be written as
\begin{equation}
w_t = m_t + n_tp_t, \quad \text{and} \quad m_t = m_{t-1} - n_t p_{t-1},\label{eq_wtmt_NoTransCosts}
\end{equation} 
with $m_t$ the cash position, $n_t$ the number of shares in the risky asset, $p_t$ the price of the asset and $w_t$ the wealth at time $t$. At every time $t$ a decision about a reallocation of investments is made. For a final time period $t=0,1,\dots,N_t$, we wish to maximize $w_{N_t}-w_0$, which can be expanded as
\begin{equation}
w_{N_t}-w_0 = \sum_{t=1}^{N_t} w_t - w_{t-1}.\label{eq_DeltawNt_NoTransCosts}
\end{equation}
To maximize \eqref{eq_DeltawNt_NoTransCosts}, we thus have to maximize the increments. Combining \eqref{eq_wtmt_NoTransCosts}, we write $
w_t-w_{t-1} = n_t(p_t-p_{t-1}) - n_{t-1}p_{t-1}$,
which therefore motivates the following optimal trading strategy, implemented at every $t-1$. If $p_t> p_{t-1}$, maximize $n_t$ and set $n_{t-1}=0$ (i.e., allocate maximal ressources towards the asset), and if $p_t\leq p_{t-1}$, set $n_t=0$ and minimize $n_{t-1}$ (i.e., sell the asset if held at $t-1$ and allocate maximal ressources towards the cash position). We profit on a price increase of the asset, and maintain our wealth on a price decrease\footnote{We here assume a \emph{long-only} strategy. By the use of derivative contracts, we can increase wealth on a price decrease as well.}. Thus, it is optimal to trade upon any change of sign of $\Delta p_t = p_t-p_{t-1}$. This is visualized Figure \ref{fig:hft_NoTransCosts} and summarized in the following remark.   
\begin{figure}[htbp]
\centering  
\begin {tikzpicture}[-latex ,auto ,node distance =4 cm and 5cm ,on grid ,
semithick ,
state/.style ={ circle ,top color =white , bottom color = white ,
draw, text=blue , minimum width =1 cm}]
\node[state] (A) {cash};
\node[state] (B) [right =of A] {asset};
\path (A) edge [loop left] node[left] {if $\Delta p_t\leq 0$} (A);
\path (A) edge [bend left =25] node[above] {if $\Delta p_t>0$} (B);
\path (B) edge [bend left =15] node[below =0.15 cm] {if $\Delta p_t\leq 0$} (A);
\path (B) edge [loop right] node[right] {if $\Delta p_t> 0$} (B);
\end{tikzpicture}
\caption{Visualization of the Markov decision process when optimally trading only cash and an asset in the absence of transaction costs. The optimal trading strategy is to trade upon any change of $\Delta p_t$-sign, i.e., even if this is minimally small ($\Delta p_t \rightarrow 0$).}
\label{fig:hft_NoTransCosts}
\end{figure}
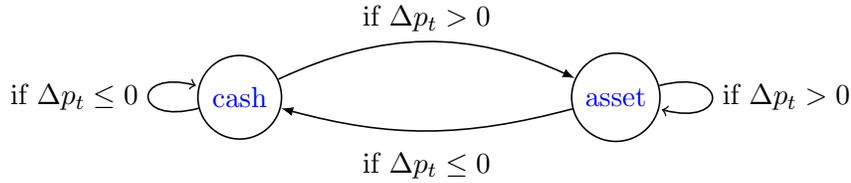
\begin{remark}
In absence of transaction costs, trading upon any change in sign of $\Delta p_t = p_t-p_{t-1}$ is the optimal trading policy.\qed \label{remark_hft}
\end{remark}
Remark \ref{remark_hft} implies that in the absence of transaction costs \emph{high-frequency trading} (\cite{aldridge2013high}, \cite{bowen2010high}) is always the optimal trading policy. Furthermore, note that $\Delta w_t = w_t-w_{t-1}\geq0,~\forall t=1,2,\dots,N_t$. Thus, in the absence of transaction costs, at every time step there is at least no incremental decrease in wealth when employing the optimal trading policy. Naturally, when including non-zero transaction costs, this is in general not the case anymore and we may (at least temporarily) have $\Delta w_t<0$. In addition, the optimal trading frequency will be non-trivially affected. Quantitative examples for optimal trading frequencies under transaction costs are given in Section \ref{sec_NumEx}.

Fifth, motivated by the previous paragraph and under the consideration of transaction costs, a valid question to address is when to sell and rebuy a non-currency asset given a long-term trend but temporary dip in price. Selling and rebuying may optimize profit. The typical minimal decrease in price required for the strategy of selling and rebuying being optimal is approximately twice the proportional transaction cost level. Twice because of selling \emph{and} rebuying. Approximately because of cash residuals due to the integer-valued number of assets and fixed transaction costs that need to be accounted for. 

\section{Multi-stage optimization without diversification constraints\label{sec_multi_stage_opt_without_diversification}}

\subsection{Multi-stage optimization}

Multi-stage transition dynamics can be modeled in form of a transition graph. We therefore assign a set $\mathcal{Z}_t$ of admissable states to every time stage $t$. For investment trajectory optimization without a diversification constraint, we employ \emph{one} transition graph. For investment trajectory optimization with a diversification constraint, \emph{multiple} transition graphs and  $\mathcal{Z}_t^{(q)},~\forall q=0,1,\dots,Q-1$, are defined and discussed in Section \ref{sec_multi_stage_opt_with_diversification}. In contrast, for the remainder of this section we dismiss superscript ``$^{(q)}$'' and focus on optimization without a diversification constrain. We define the initial set $\mathcal{Z}_0 = \left\{\bm{z_0}:\bm{z_0} = \begin{bmatrix} 0 & 0 & 0 & m_0^{0} & 0 & m_0^0 & 0 & 0 \end{bmatrix} \right\}$. In the following, three constraints are discussed that affect transition graph generation.

\subsection{Case 1: Unconstrained trading frequency}

\begin{remark}
Suppose that following a particular investment trajectory, at time $\tau$ an investment state $z_\tau$ is reached with a particular $i_\tau$, $w_\tau^0$ and $j_\tau=i_{\tau-1}$. Suppose further that there exists another investment trajectory resulting in the same asset, i.e., $\tilde{i}_\tau=i_\tau$, but in contrast with $\tilde{w}_\tau^0>w_\tau^0$ and $\tilde{j}_\tau\neq j_\tau$. Then, the former investment trajectory can be dismissed from being a possible candidate segment for the optimal investment trajectory. This is because any trajectory continuing the latter investment trajectory will always outperform a continuation of the former investment trajectory for all $t>\tau$. \qed 
\label{remark1}      
\end{remark}
Remark \ref{remark1} motivates a simple but efficient transition graph generation: first, branch from every state $z_{t-1}\in\mathcal{Z}_{t-1}$ to all possible states $z_t$ at time $t$ according transition dynamics \eqref{eq_def_zt}, whereby we summarize the set of states at time $t-1$ from which $z_t$ can be reached as $\mathcal{J}_{t-1}^{z_t}$; second, select the optimal transitions and thus determine $\mathcal{Z}_t$ according to
\begin{equation}
\mathcal{Z}_t = \left\{\bm{z_t}: \max_{j_t\in\mathcal{J}_{t-1}^{z_t}} \{w_t^0\},~\forall i_t\in\mathcal{I}\right\},\label{eq_case1_nodiversific_mZt}
\end{equation}
recalling the definition $j_t=i_{t-1}$, and thereby selecting the solutions with highest value $w_t^0$, $\forall i_t\in\mathcal{I}= \{ 0,1,\dots,N_c+N_a-1\}$. The resulting transition graph holds a total of $N_z(t) = 1+(N_c+N_a)t$ states up to time $t\geq 0$. For a time horizon $N_t$, the optimal investment strategy, here denoted by superscript ``$^\star$'', can then be reconstructed by proceeding backwards as 
\begin{equation}
\begin{aligned}
\bm{z_{N_t}^\star} &= \left\{ \bm{z_{N_t}}: i_{N_t}^\star = \max_{i_{N_t}} \{w_{N_t}^0 \},~i_{N_t}\in\mathcal{I} \right\},\\
\bm{z_{t-1}^\star} &= \left\{ \bm{z_{t-1}} : i_{t-1} = j_t^\star\right\}, \quad \forall t=N_t,N_t-1,\dots,1.
\end{aligned}\label{eq_zstar_reconstruction}
\end{equation}
The resulting investment trajectory is optimal since by construction of the transition graph as outlined, starting from $z_0$, there exists exactly one wealth maximizing trajectory to every investment $i_t=0,1,\dots,N_c+N_a-1$ for every time $t=0,1,\dots,N_t$. By iterating backwards the optimal investment decisions at every time stage are determined.  

\subsection{Case 2: Bound on the admissable number of trades\label{subsec_case2_NoDiversificationCstrt}}

We constrain the investment trajectory to include at most $K\in\mathbb{Z}_+$ trades during $t=0,1,\dots,N_t$, whereby we define a \emph{trade} as any reallocation of an investment resulting in a change of the asset identification number $i_t$. A transition according $i_t=i_{t-1}$ is consequently no trade. The set of admissable states is generated as 
\begin{equation}
\mathcal{Z}_t = \left\{ \bm{z_t}: \max_{j_t\in\mathcal{J}_{t-1}^{z_t}} \{ w_t^0\},~\forall k_t<K~\text{ and unique}, \text{and}~\forall i_t\in\mathcal{I} \right\}.\label{eq_mathcalZt_case2}
\end{equation}
Consequently, the resulting transition graph holds a total of 
$ N_z(t) = 1 + \sum_{l=1}^t (N_c+N_a)\text{min}(l,K)$ states. The reconstruction of optimal investment decisions is similar to \eqref{eq_zstar_reconstruction}. 

Note that the total number of states, $N_z(t)$, quickly reaches large numbers. We therefore introduce a heuristic to reduce $N_z(t)$ while not compromising optimality of the solution. 
\begin{proposition}
While not compromising the finding of an optimal investment trajectory, the set of admissable states $\mathcal{Z}_t$ of \eqref{eq_mathcalZt_case2} can be shrunken to $\tilde{\mathcal{Z}}_t$ according to the following heuristic:
%
\begin{algorithmic}[1]
\vspace{-0cm}\State \parbox[t]{\dimexpr\linewidth-\algorithmicindent}{Initialize: $\tilde{\mathcal{Z}}_t= \mathcal{Z}_t$.\strut}
\NoDo\NoThen
\vspace{0.05cm}\For{\parbox[t]{\dimexpr\linewidth-\algorithmicindent}{every $i_t\in\mathcal{I}$ such that the corresponding $\bm{z_t}\in\mathcal{Z}_t$ of \eqref{eq_mathcalZt_case2}:\strut}}
\vspace{-0.35cm}\State \parbox[t]{\dimexpr\linewidth-\algorithmicindent-\algorithmicindent}{Compute: $k_t^\text{opt}(i_t) = \left\{k_t: w_t^{0,\text{opt}}(i_t)=\max\{w_t^0\}~\text{s.t. corresponding}~\bm{z_t}\in\mathcal{Z}_t~\text{of \eqref{eq_mathcalZt_case2}} \right\}$.\strut}
\vspace{-0.cm}\State \parbox[t]{\dimexpr\linewidth-\algorithmicindent-\algorithmicindent}{Shrink: $\tilde{\mathcal{Z}}_t = \tilde{\mathcal{Z}}_t \backslash \left\{ \bm{z_t}: k_t > k_t^\text{opt}(i_t) \right\}$.\strut}
\EndFor
\end{algorithmic}
%
\end{proposition}
\begin{proof}
W.l.o.g., suppose that for a given $i_t=i\in\mathcal{I}$ we have determined $k_t^\text{opt}(i_t)$. Let the associated state vector be denoted by $\bm{z_t^\text{opt}(i_t)}$. Then, we can discard all $\bm{z_t}$ with $i_t=i$ and $k_t>k_t^\text{opt}(i_t)$, since $w_{t+\tau}^{0,\text{opt}}(i_t)\geq w_{t+\tau}^0,~\forall \tau\geq 0$, and the admissable set for state $\bm{z_t^\text{opt}(i_t)}$ is thus larger by at least the option of one additional trade, in comparison to the admissable set corresponding to all $\bm{z_t}\in\mathcal{Z}_t$ of \eqref{eq_mathcalZt_case2} with $i_t=i$ and $k_t>k_t^\text{opt}(i_t)$. 
\end{proof}

Note that the total number of states, $N_z(N_t)$, cannot be predicted precisely as before. It is now data-dependent instead. Quantitative implications are reported in Section \ref{sec_NumEx}. 

\subsection{Case 3: Waiting period after every trade until the next trade \label{subsec_case3_NoDiversificationCstrt}}

We constrain the investment trajectory to waiting of at least a specific time period $D$ after every executed trade until the next trade. The set of admissable states is consequently generated as 
\begin{equation}
\mathcal{Z}_t = \left\{ \bm{z_t}: \max_{j_t\in\mathcal{J}_{t-1}^{z_t}} \{ w_t^0\},~\forall d_t<D~\text{ and unique}, \text{and}~\forall i_t=\in\mathcal{I} \right\}.\label{eq_mathcalZt_case3}
\end{equation} 
As a result, the resulting transition graph holds a total of 
$ N_z(t) = 1 + \sum_{l=1}^t  (N_c+N_a-1)\text{min}(l,D) + 1 + \text{min}\left( \text{max}(0,l-D),D-1\right) $ states. The reconstruction of optimal investment decisions is similar to \eqref{eq_zstar_reconstruction}. 

Similarly to Section \ref{subsec_case2_NoDiversificationCstrt}, the total number of states, $N_z(t)$, quickly reaches large numbers. We therefore also introduce a heuristic to reduce $N_z(t)$ while not compromising optimality of the solution. 
\begin{proposition}
While not compromising the finding of an optimal investment trajectory, the set of admissable states $\mathcal{Z}_t$ of \eqref{eq_mathcalZt_case3} can be shrunken to $\bar{\mathcal{Z}}_t$ according to the following heuristic:
%
\begin{algorithmic}[1]
\vspace{-0cm}\State \parbox[t]{\dimexpr\linewidth-\algorithmicindent}{Initialize: $\bar{\mathcal{Z}}_t= \mathcal{Z}_t$.\strut}
\NoDo\NoThen
\vspace{0.05cm}\For{\parbox[t]{\dimexpr\linewidth-\algorithmicindent}{every $i_t\in\mathcal{I}$ such that the corresponding $\bm{z_t}\in\mathcal{Z}_t$ of \eqref{eq_mathcalZt_case3}:\strut}}
\vspace{-0.35cm}\State \parbox[t]{\dimexpr\linewidth-\algorithmicindent-\algorithmicindent}{Compute: $d_t^\text{opt}(i_t) = \left\{d_t: w_t^{0,\text{opt}}(i_t):=\max\{w_t^0\}~\text{s.t. corresponding}~\bm{z_t}\in\mathcal{Z}_t~\text{of \eqref{eq_mathcalZt_case3}} \right\}$.\strut}
\vspace{-0.cm}\State \parbox[t]{\dimexpr\linewidth-\algorithmicindent-\algorithmicindent}{Shrink: $\bar{\mathcal{Z}}_t = \bar{\mathcal{Z}}_t \backslash \left\{ \bm{z_t}: 0<d_t < d_t^\text{opt}(i_t) \right\}$.\strut}
\EndFor
\end{algorithmic}
%
\end{proposition}
\begin{proof}
W.l.o.g., suppose for a given $i_t=i\in\mathcal{I}$ we have determined $d_t^\text{opt}(i_t)$. Let the associated state vector be denoted by $\bm{z_t^\text{opt}(i_t)}$. Then, we can discard all $\bm{z_t}$ with $i_t=i$ and $0<d_t<d_t^\text{opt}(i_t)$, since $w_{t+\tau}^{0,\text{opt}}(i_t)\geq w_{t+\tau}^0,~\forall \tau\geq 0$, and the admissable set for state $\bm{z_t^\text{opt}(i_t)}$ is larger by being closer to a potential next trade by at least one trading sampling time, in comparison to the admissable set corresponding to all $\bm{z_t}\in\mathcal{Z}_t$ of \eqref{eq_mathcalZt_case3} with $i_t=i$ and $0<d_t<d_t^\text{opt}(i_t)$.
\end{proof}

Similarly to Section \ref{subsec_case2_NoDiversificationCstrt}, the total number of states, $N_z(N_t)$, cannot be predicted precisely since it is data-dependent. Quantitative results are reported in Section \ref{sec_NumEx}. This heuristic significantly reduces computational complexity in practice. 

\section{Multi-stage transition dynamics optimization with a diversification constraint\label{sec_multi_stage_opt_with_diversification}}

In portfolio optimization the introduction of diversification constraints is regarded as a measure to reduce drawdown risk. For our purpose of analysis of historical optimal trading we first divide the initial wealth $m_0$ into $Q$ parts of equal proportion. Then, we impose constraints on each of the corresponding $Q$ investment trajectories. In the unconstrained case, all $Q$ trajectories would coincide. In the constrained case, we distinguish between i) constraints \emph{between} multiple investment trajectories: diversification at only the initial time, diversification permitted at all times, asynchronous trading and synchronous trading, and ii) constraints \emph{along} any specific investment trajectory: unconstrained trading frequency, at most $K$ trades along the investment trajectory, and the enforcement of a waiting period after each executed trade. 

We define a diversification constraint at a specific time $t$ such that each of the states of the $Q$-trajectories, $\bm{z_t}\in\mathcal{Z}_t^{(q)}$, $\forall q=0,\dots,Q-1$, must be invested differently. Thus, each asset identification number $i_t^{(q)}$ must be different $\forall t=0,1,\dots,N_t$, $\forall q=0,1,\dots,Q-1$. 

We define the sets of admissable states $\mathcal{Z}_t^{(q)}$, $\forall t=0,1,\dots,N_t$ and $\forall q=0,1,\dots,Q-1$, sequentially and ordered according to optimality. Thus, $\mathcal{Z}_t^{(1)}$, $\forall t=0,1,\dots,N_t$ is constructed accounting only for the optimal investment trajectory associated with $\mathcal{Z}_t^{(0)}$, i.e., the set $\mathcal{Z}_t^{(0),\star},~\forall t=0,1,\dots,N_t$, whereas $\mathcal{Z}_t^{(q)}$ is constructed accounting for all of the optimal investment trajectories associated with $\mathcal{Z}_t^{(0),\star},~\mathcal{Z}_t^{(1),\star},\dots,\mathcal{Z}_t^{(q-1),\star}$. Here, $\mathcal{Z}_t^{(q),\star},~\forall q=0,1,\dots,Q-1$ denotes the set of states at each time $t$ that  result from the reconstruction of optimal investment decisions along the optimal investment trajectory according to \eqref{eq_zstar_reconstruction}. Thus, our methodology aims at being maximally invested in the investment trajectories ordered according to optimality.

\subsection{Q trajectories, diversification for a subset of times and asynchronous trading}

We define the subset of trading sampling times as $\mathcal{T}^{(q)}\subseteq \{0,1,\dots,N_t\},~\forall q=0,1,\dots,Q-1$. For enforcement of diversification in form of $Q$ trajectories, diversification for any subset of trading times and asynchronous trading, the sets of admissable states are initialized as
\begin{equation}
\begin{aligned}
\mathcal{Z}_0^{(q)} = \left\{\bm{z_0}:\bm{z_0} = \begin{bmatrix} 0 & 0 & 0 & m_0^{0,q} & 0 & m_0^{0,q} & 0 & 0 \end{bmatrix} \right\},~\forall q=0,1,\dots,Q-1.
\end{aligned}\label{eq_def_Z0q_asyn}
\end{equation}

For unconstrained trading frequency along an investment trajectory and $t>0$, the sets of admissable states are thus generated according to
\begin{equation}
\begin{aligned}
\mathcal{Z}_t^{(q)} & = \biggg\{\bm{z_t}: \max_{j_t\in\mathcal{J}_{t-1}^{z_t}} \{w_t^0\},~\forall i_t\in\mathcal{I}~\text{if}~t\notin\mathcal{T}^{(q)},~\text{or}~...\\
& \hspace{0cm} \forall i_t\in\mathcal{I}\backslash\left( \cup \left\{ i_t: i_t = i_t^{(r),\star}, \bm{z_t^{(r),\star}}\in\mathcal{Z}_t^{(r),\star} \right\}_{r=0}^{q-1} \right)~\text{if}~ t\in\mathcal{T}^{(q)}\biggg\},
\end{aligned}
\label{eq_Ztq_case1_asyn}
\end{equation} 
with $q=0,1,\dots,Q-1$, and where $\bm{z_t^{(r),\star}}\in\mathcal{Z}_t^{(r)}$ denotes the optimal state at time $t$ associated with investment trajectory $r$. 

For the case of at most $K$ admissable trades along any investment trajectory and $t>0$, the sets of admissable states are generated according to 
\begin{equation}
\begin{aligned}
\mathcal{Z}_t^{(q)} &= \biggg\{\bm{z_t}: \max_{j_t\in\mathcal{J}_{t-1}^{z_t}} \{w_t^0\},~\forall k_t<K~\text{ and unique},~\text{and}~\forall i_t\in\mathcal{I}~\text{if}~t\notin\mathcal{T}^{(q)},~\text{or}~...\\
& \hspace{-0.5cm} \forall k_t<K~\text{ and unique},~\text{and}~\forall i_t\in\mathcal{I}\backslash \left( \cup \left\{ i_t: i_t = i_t^{(r),\star}, \bm{z_t^{(r),\star}}\in\mathcal{Z}_t^{(r),\star} \right\}_{r=0}^{q-1} \right)~\text{if}~ t\in\mathcal{T}^{(q)} \biggg\},
\end{aligned}
\label{eq_Ztq_case2_asyn}
\end{equation} 
with $q=0,1,\dots,Q-1$.

For the case of enforcing a waiting period after each executed trade along any investment trajectory and $t>0$, the sets of admissable states are generated according to
\begin{equation}
\begin{aligned}
\mathcal{Z}_t^{(q)} &= \biggg\{\bm{z_t}: \max_{j_t\in\mathcal{J}_{t-1}^{z_t}} \{w_t^0\},~\forall d_t<D~\text{ and unique},~\text{and}~\forall i_t\in\mathcal{I}~\text{if}~t\notin\mathcal{T}^{(q)},~\text{or}~...\\
& \hspace{-0.5cm} \forall d_t<D~\text{ and unique},~\text{and}~\forall i_t\in\mathcal{I}\backslash \left( \cup \left\{ i_t: i_t = i_t^{(r),\star}, \bm{z_t^{(r),\star}}\in\mathcal{Z}_t^{(r),\star} \right\}_{r=0}^{q-1} \right)~\text{if}~ t\in\mathcal{T}^{(q)} \biggg\},
\end{aligned}
\label{eq_Ztq_case3_asyn}
\end{equation} 
with $q=0,1,\dots,Q-1$.

\subsection{Q trajectories, diversification for all times and synchronous trading}

Let us define a subset of trading sampling times as $\mathcal{T}\subseteq \{0,1,\dots,N_t\}$. This subset may, for example, indicate the sampling times at which trades were executed along the optimal investment trajectory associated with $\mathcal{Z}_t^{(0),\star}$:
\[
\mathcal{T} = \{t:i_t^{(0),\star}\neq j_t^{(0),\star},~\bm{z_t^{(0),\star}}\in\mathcal{Z}_t^{(0),\star},~\forall t=1,\dots,N_t \}.
\]

The set of a admissable states is initialized as in \eqref{eq_def_Z0q_asyn}. Then, for an unconstrained trading frequency along an investment trajectory and $t>0$, the sets of admissable states are generated according  to 
\begin{equation}
\begin{aligned}
\mathcal{Z}_t^{(q)}  = \biggg\{\bm{z_t}:\hspace{0.1cm} & \bm{z_t}=\bm{z_{t-1}}~\text{if}~ t\notin\mathcal{T},~\text{or}~\dots\\
& \hspace{0cm} \bm{z_t}~\text{s.t.}~\max_{j_t\in\mathcal{J}_{t-1}^{z_t}} \{w_t^0\},~\forall i_t\in\mathcal{I}\backslash\left( \cup \left\{ i_t: i_t = i_t^{(r),\star}, \bm{z_t^{(r),\star}}\in\mathcal{Z}_t^{(r),\star} \right\}_{r=0}^{q-1} \right)~\text{if}~t\in\mathcal{T}
\biggg\},
\end{aligned}
\label{eq_Ztq_case1_syn}
\end{equation} 
with $q=0,1,\dots,Q-1$. 

The case of at most $K$ admissable trades along any investment trajectory as well as the case of enforcing a waiting period after each executed trade along any investment trajectory can then be defined analogously.

\subsection{Remarks and relevant quantities for interpretation\label{subsec_remarks_evalQuantities}}

Note that the presented framework can also be extended to analyze alternative optimization criteria such as, for example, determining a worst-case investment trajectory (\emph{pessimization}), or the tracking of a target return reference trajectory (\emph{index tracking}).

In order to interpret quantiative results in the following Section \ref{sec_NumEx}, we define the \emph{total return} (measured in percent) as $r_{N_t}^{\text{tot},(q)} = 100\frac{w_{N_t}^0-w_0^0}{w_0^0}, \quad \forall q \in \mathcal{Q}$. Similarly, we define the return at time $t$ as $r_{t}^{\text{tot},(q)},~\forall q\in\mathcal{Q}$. We further report the total number of conducted trades as $K_{N_t}^{\text{tot}}$. The minimal time-span between any two trades within time-frame $t\in\mathcal{N}_t=\{0,1,\dots,N_t\}$ shall be denoted by $D_{N_t}^\text{min}$. In addition, the average, minimal and maximal percentage gain per conducted non-currency asset-trade is of our interest. Stating the quantities with respect to our reference currency (EUR), we therefore first define the set
\begin{equation*}
\begin{aligned}
\Delta \mathcal{G}^{(q)} &= \bigg\{ 100 \frac{w_{\tau}^0 - w_{\eta}^{0} }{w_{\eta}^{0}} : \text{with}~\tau~\text{s.t.}~\tau = t-1,~\bar{i}\in\mathcal{I}_{N_a},~i_{t-1}\neq \bar{i},~i_t=\bar{i}, \\
&\text{with}~\eta~\text{s.t.}~\eta=t,~\bar{i}\in\mathcal{I}_{N_a}, ~i_t=\bar{i},~i_{t+1}\neq \bar{i},~\text{and}~\tau>\eta,~\bm{z_t}\in\mathcal{Z}_t^{(q),\star},~\forall t\in\mathcal{N}_t,~\forall q\in\mathcal{Q} \bigg\},
\end{aligned}
\end{equation*}
whereby $\bar{i}$ identifies an asset of interest. The average, minimum and maximum shall then be denoted by $\text{avg}(\Delta \mathcal{G}^{(q)})$, $\text{min}(\Delta \mathcal{G}^{(q)})$ and $\text{max}(\Delta \mathcal{G}^{(q)})$, respectively. The associated trading times are summarized in 
\begin{equation*}
\begin{aligned}
\Delta \mathcal{T}^{(q)} &= \bigg\{ \tau-\eta : \text{with}~\tau~\text{s.t.}~\tau = t-1,~\bar{i}\in\mathcal{I}_{N_a},~i_{t-1}\neq \bar{i}, ~i_t=\bar{i}, \\
&\text{with}~\eta~\text{s.t.}~\eta=t,\bar{i}\in\mathcal{I}_{N_a},~i_t=\bar{i},~i_{t+1}\neq \bar{i},~\text{and}~\tau>\eta,~\bm{z_t}\in\mathcal{Z}_t^{(q),\star},~\forall t\in\mathcal{N}_t,~\forall q\in\mathcal{Q} \bigg\},
\end{aligned}
\end{equation*}
with corresponding $\text{avg}(\Delta \mathcal{T}^{(q)})$, $\text{min}(\Delta \mathcal{T}^{(q)})$ and $\text{max}(\Delta \mathcal{T}^{(q)})$ defined accordingly.

Then, we can partition quantities of interest into two groups: overall performance measures and, secondly, quantities associated with non-currency asset holdings along an investment-optimal $q$-trajectory. We therefore compactly summarize results in evaluation vectors and matrices 
\begin{align}
\bm{e}^{(q)} &= \begin{bmatrix} r_{N_t}^{\text{tot},(q)} &  K_{N_t}^{\text{tot},(q)} & D_{N_t}^{\text{min},(q)} \end{bmatrix},\quad \forall q\in\mathcal{Q},\\
\bm{E}^{(q)} &= \begin{bmatrix}\text{avg}(\Delta \mathcal{G}^{(q)}) & \text{min}(\Delta \mathcal{G}^{(q)}) & \text{max}(\Delta \mathcal{G}^{(q)}) \\ \text{avg}(\Delta \mathcal{T}^{(q)}) & \text{min}(\Delta \mathcal{T}^{(q)}) & \text{max}(\Delta \mathcal{T}^{(q)}) \end{bmatrix},\quad \forall q\in\mathcal{Q}.
\end{align}

\section{Numerical examples\label{sec_NumEx}}

To quantitatively evaluate results, three numerical examples are reported. For all examples, a time horizon of one year is chosen. The sampling time is selected as one day. Adjusted closing prices of both foreign exchange rates and stock indices are retrieved from \texttt{finance.yahoo.com}. As a preprocessing step all non-currency assets are normalized to value $100$ in their corresponding currency at time $t=0$. 

The first example treats optimal trading of EUR, USD and the Nasdaq-100. This scenario is selected mainly to analyze currency effects. No diversification constraint is enforced such that we have $Q=1$. The second example treats optimal trading of 16 different currencies and 15 different non-currency assets. A diversification constraint is employed with $Q=3$. The third example compares the results for an exemplary downtrending and an uptrending stock.   

This section illustrates the effects of i) different transaction cost levels, and ii) various constraints on a posteriori optimal trading performance. 

\subsection{Example 1: EUR, USD and Nasdaq-100\label{subsec_ex1}}

\begin{table}
\begin{center}
\tbl{Example 1, see Section \ref{subsec_ex1}. Identification of the currencies and asset under consideration. The currency in which asset $i$ is traded is denoted by $c(i)$. The reference currency of Nasdaq-100 is USD.\label{tab1_ex1}}
{
\def\arraystretch{0.7}
\begin{tabular}{@{}lccc}\toprule
$i$ & \texttt{finance.yahoo}-symbol & Interpretation & $c(i)$ \\\colrule
$0$ & -- & EUR & $0$ \\[0.1cm]
$1$ & \texttt{EURUSD=x} & USD & $1$ \\[0.1cm]
$2$ & \texttt{\^~NDX} & Nasdaq-100 & $1$ \\[0.1cm]
\botrule
\end{tabular}}
\end{center}
\end{table}

\begin{table}
\begin{center}
\tbl{Summary of quantitative results of Example 1 from Section \ref{subsec_ex1}.\label{tab_ex1}}
{
\def\arraystretch{0.7}
\begin{tabular}{@{}lcccc}\toprule
 & $\epsilon=0$ & $\epsilon=0.5$
& $\epsilon=1$
& $\epsilon=2$ \\\colrule
Buy-and-Hold & $\begin{bmatrix} \bm{-0.2}&1&250 \end{bmatrix}$ & $\begin{bmatrix} \bm{-1.2}&1&250 \end{bmatrix}$ & $\begin{bmatrix} \bm{-2.2}&1&250 \end{bmatrix}$ & $\begin{bmatrix} \bm{-4.2}&1&250 \end{bmatrix}$ \\[0.2cm]
Unconstrained & $\begin{bmatrix} \bm{330.0}&165&1 \end{bmatrix} \atop \begin{bmatrix} 2.0&0.1&11.0\\1.8&1&5 \end{bmatrix}$  & $\begin{bmatrix} \bm{134.6}&59&1 \end{bmatrix} \atop \begin{bmatrix} 3.9&0.9&11.8\\5.5&1&21 \end{bmatrix}$ & $\begin{bmatrix} \bm{86.2}&31&1 \end{bmatrix} \atop \begin{bmatrix} 5.9&0.4&16.6\\12.2&1&45 \end{bmatrix}$ & $\begin{bmatrix} \bm{50.1}&14&1 \end{bmatrix} \atop \begin{bmatrix} 8.2&3.1&15.5\\20.0&3&45 \end{bmatrix}$  \\[0.7cm]
$\leq 12$ trades & $\begin{bmatrix} \bm{106.8}&12&6 \end{bmatrix} \atop \begin{bmatrix} 11.9&7.1&17.8\\23.7&11&45 \end{bmatrix}$ & $\begin{bmatrix} \bm{87.9}&12&6 \end{bmatrix} \atop \begin{bmatrix} 11.0&6.6&17.2\\23.7&11&45 \end{bmatrix}$ & $\begin{bmatrix} \bm{72.5}&12&5 \end{bmatrix} \atop \begin{bmatrix} 10.1&5.8&16.6\\23.8&11&45 \end{bmatrix}$ & $\begin{bmatrix} \bm{49.6}&12&1 \end{bmatrix} \atop \begin{bmatrix} 9.2&5.0&15.5\\23.7&11&45 \end{bmatrix}$  \\[0.7cm]
$\geq 10$ days waiting & $\begin{bmatrix} \bm{114.2}&18&10 \end{bmatrix} \atop \begin{bmatrix} 9.1&2.7&16.6\\15.3&10&21 \end{bmatrix}$ & $\begin{bmatrix} \bm{84.8}&17&10 \end{bmatrix} \atop \begin{bmatrix} 9.2&2.7&16.0\\ 16.1&11&21 \end{bmatrix}$ & $\begin{bmatrix} \bm{68.4}&13&10 \end{bmatrix} \atop \begin{bmatrix} 9.8&6.0&15.4\\ 22.3&11&45 \end{bmatrix}$ & $\begin{bmatrix} \bm{46.5}&10&11 \end{bmatrix} \atop \begin{bmatrix} 9.3&5.9&15.5\\ 25.4&12&45 \end{bmatrix}$ \\[0.55cm]
\botrule
\end{tabular}}
\end{center}
\end{table}

The results for numerical example 1 are summarized in Table \ref{tab_ex1}. Different levels of transaction costs with variable proportional cost but constant fixed cost are considered. In Table \ref{tab_ex1} the evaluation quantities $\bm{e^{(0)}}$ and $\bm{E^{(0)}}$ are reported for the different trading strategies. For the Buy-and-Hold strategy, only $\bm{e^{(0)}}$ is reported. We assume proportional costs (indicated in $\%$) to be the same for buying and selling for both foreign exchange and asset trading, i.e., $\epsilon = \epsilon_\text{buy}=\epsilon_\text{sell}$. For $\epsilon=0$, we also set $\beta=0$. For all other cases, we set $\beta=50$. Total returns ($r_{N_t}^{\text{tot},(q)}$) are printed bold for emphasis. The time-span of interest is August 5, 2015 until August 3, 2016, and comprises 251 potential trading days. 

Several observations can be made with respect to the results of Table \ref{tab_ex1}. First, eventhough only two currencies, EUR and USD, i.e., $i\in\mathcal{I}_{N_c}=\{0,1\}$, and one non-currency asset, i.e., $i\in\mathcal{I}_{N_a}=\{2\}$, are traded long-only, remarkable profits can be earned when optimally trading a posteriori. Even in case of (high) transaction costs with a proportional rate of $2\%$, the profits significantly outperform a one-year Buy-and-Hold strategy. Second, the influence of different levels of transaction costs is impressive. This holds specifically for unconstrained trading with respect to returns, optimal trading frequency and percentage gains (average, minimum and maximum) upon which the non-currency asset is traded. Third, while the total return drops with increasing transaction cost levels, the remaining evaluation quantities remain approximately constant for the $K$-trades strategy (here $K=12$, i.e., 12 trades per year or one per month). Fourth, the results associated with the percentage gains upon which the non-currency asset is traded were unexpected. Intuitively, they were thought to be higher. The same holds for optimal time periods between any two trades. Results from Example 1 encourage frequent trading. For example, for the case with a waiting constraint, trading is encouraged upon percentage gains of on average slightly less than $10\%$ for all four levels of transaction costs.

\newlength\figureheight
\newlength\figurewidth
\setlength\figureheight{5cm}
\setlength\figurewidth{15cm}
\begin{figure}
\centering
\input{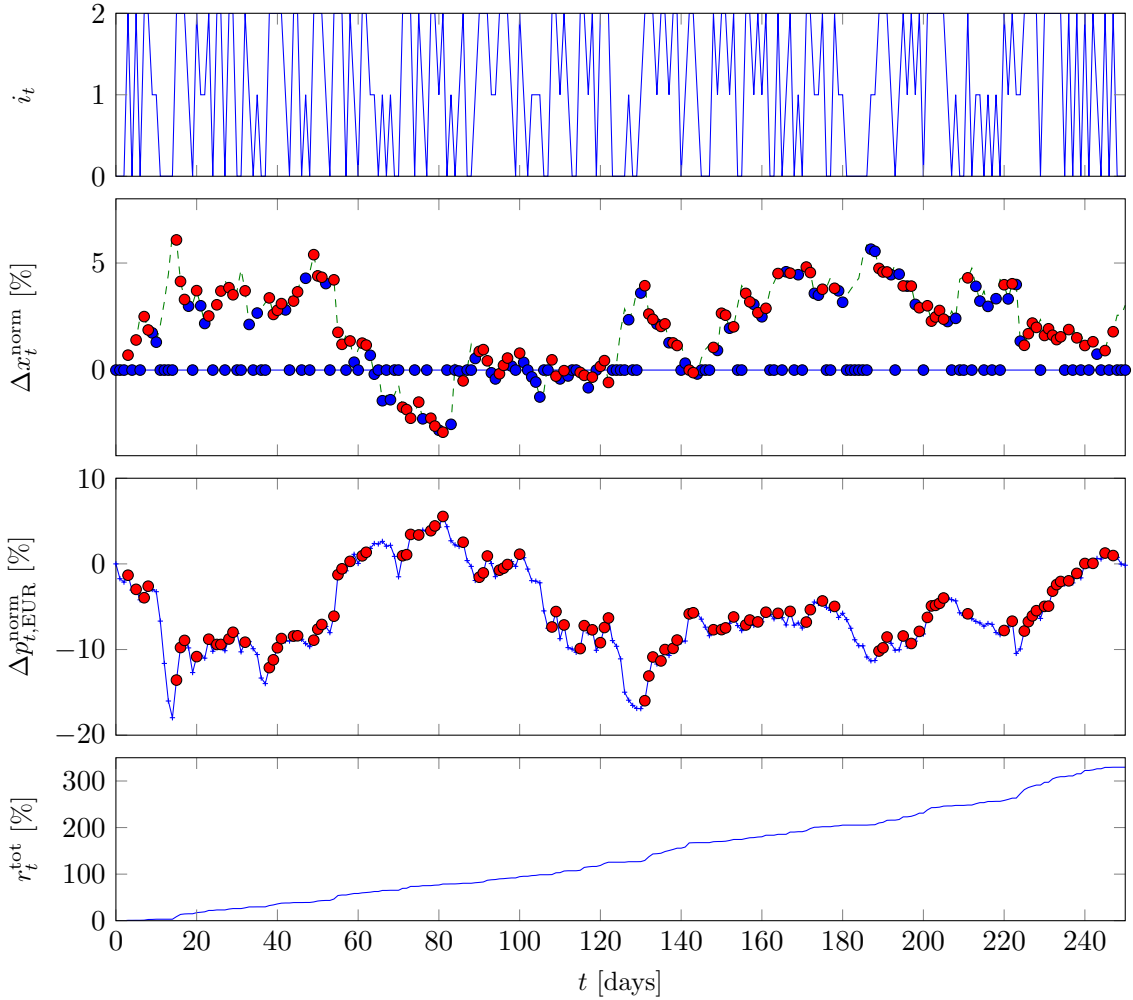}
\caption{Example 1, the unconstrained trading case in the absence of any transaction costs. See also Table \ref{tab_ex1}.}
\label{fig:ex1_fig1}
\end{figure}

Figure \ref{fig:ex1_fig1} further visualizes results. In order to compactly display multiple foreign exchange rates, we normalize w.r.t. the initial value at $t=0$, see the subplot with label $\Delta x_t^\text{norm}$. For reference currency EUR, we set $\Delta x_t^\text{norm}=0,~\forall t=0,1,\dots,N_t$. Analogously, we normalize non-currency prices and additionally take currency effects into account by first converting prices to currency EUR, see the subplot with label $\Delta p_{t,\text{EUR}}^\text{norm}$. For a specific optimal investment trajectory, at every time $t$, an investment in exactly one currency or non-currency asset is taken. Being invested in a non-currency asset is indicated by the red balls in Figure \ref{fig:ex1_fig1}. Since non-currency assets are associated with a specific currency we also label them accordingly with red balls. In contrast, an explicit investment in a currency is emphasized by blue balls.

It is striking that despite an absence of clear trends in both the EURUSD-foreign exchange rate and the Nasdaq-100 stock index, significant profits can be made when optimally trading--even when employing a \emph{long-only} strategy. The largest increases in return rates in currency EUR are achieved when the asset is increasing in value while the foreign exchange rate with reference Euro is decreasing. Investments in USD are optimal when the EURUSD-foreign exchange rate is trending down and the Nasdaq-100 is decreasing likewise. Investments in EUR are in general optimal when the EURUSD-foreign exchange rate is trending up and the Nasdaq-100 is trending down.  

\subsection{Example 2: Global investing and including a diversification constraint\label{subsec_ex2}}

\begin{table}
\begin{center}
\tbl{Example 2, see Section \ref{subsec_ex2}. Identification of 16 currencies and 15 assets. Each currency is associated with a foreign exchange rate with respect to EUR. The currency in which an asset $i$ is traded is denoted by $c(i)$.\label{tab1_ex2}}
{
\def\arraystretch{0.7}
\begin{tabular}{@{}lccc}\toprule
$i$ & \texttt{finance.yahoo}-symbol & Interpretation & $c(i)$ \\\colrule
$0$ & -- & EUR & $0$ \\[0.1cm]
$1$ & \texttt{EURUSD=x} & USD & $1$ \\[0.1cm]
$2$ & \texttt{EURJPY=x} & JPY & $2$ \\[0.1cm]
$3$ & \texttt{EURGBP=x} & GBP & $3$ \\[0.1cm]
$4$ & \texttt{EURCHF=x} & CHF & $4$ \\[0.1cm]
$5$ & \texttt{EURCNY=x} & CNY & $5$ \\[0.1cm]
$6$ & \texttt{EURDKK=x} & DKK & $6$ \\[0.1cm]
$7$ & \texttt{EURHKD=x} & HKD & $7$ \\[0.1cm]
$8$ & \texttt{EURNOK=x} & NOK & $8$ \\[0.1cm]
$9$ & \texttt{EURRUB=x} & RUB & $9$ \\[0.1cm]
$10$ & \texttt{EURBRL=x} & BRL & $10$ \\[0.1cm]
$11$ & \texttt{EURAUD=x} & AUD & $11$ \\[0.1cm]
$12$ & \texttt{EURCAD=x} & CAD & $12$ \\[0.1cm]
$13$ & \texttt{EURTRY=x} & TRY & $13$ \\[0.1cm]
$14$ & \texttt{EURZAR=x} & ZAR & $14$ \\[0.1cm]
$15$ & \texttt{EURINR=x} & INR & $15$ \\[0.1cm]
$16$ & \texttt{\^~GDAXI} & DAX (GER) & $0$ \\[0.1cm]
$17$ & \texttt{FTSEMIB.MI} & FTSEMIB (ITA) & $0$ \\[0.1cm]
$18$ & \texttt{\^~OSEAX} & OSEAX (NOR) & $8$ \\[0.1cm]
$19$ & \texttt{CSSMI.SW} & SMI (SUI) & $4$ \\[0.1cm]
$20$ & \texttt{\^~NDX} & Nasdaq-100 (USA) & $1$ \\[0.1cm]
$21$ & \texttt{\^~GSPC} & S\&P 500 (USA) & $1$ \\[0.1cm]
$22$ & \texttt{\^~N225} & NIKKEI 225 (JPN) & $2$ \\[0.1cm]
$23$ & \texttt{\^~HSI} & Hang-Seng (HKG) & $7$ \\[0.1cm]
$24$ & \texttt{\^~BVSP} & IBOVESPA (BRA) & $10$ \\[0.1cm]
$25$ & \texttt{\^~AORD} & All Ordinaries (AUS) & $11$ \\[0.1cm]
$26$ & \texttt{\^~GSPTSE} & S\&P/TSX (CAN) & $12$ \\[0.1cm]
$27$ & \texttt{AFS.PA} & FTSE/JSE (RSA) & $14$ \\[0.1cm]
$28$ & \texttt{RUS.PA} & Dow Jones Russia (RUS) & $0$ \\[0.1cm]
$29$ & \texttt{INR.PA} & MSCI India (IND) & $0$ \\[0.1cm]
$30$ & \texttt{TUR.PA} & Dow Jones Turkey (TUR) & $0$ \\[0.1cm]
\botrule
\end{tabular}}
\end{center}
\end{table}

\begin{table}
\begin{center}
\tbl{Summary of quantitative results of Example 2 for the first case: time-asynchronous trading with diversification for all times. \label{tab2_ex2}}
{
\def\arraystretch{0.7}
\begin{tabular}{@{}llcccc}\toprule
& & $\epsilon=0$ & $\epsilon=0.5$ & $\epsilon=1$ & $\epsilon=2$ \\\colrule
\multirow{4}{*}{\begin{turn}{90} \hspace{-2cm} \large{$q=0$} \end{turn}} & 
Buy-and-Hold & $\begin{bmatrix} \bm{18.2}&1&199 \end{bmatrix}$ & $\begin{bmatrix} \bm{17.0}&1&199 \end{bmatrix}$ & $\begin{bmatrix} \bm{15.9}&1&199 \end{bmatrix}$ & $\begin{bmatrix} \bm{13.6}&1&199 \end{bmatrix}$ \\[0.2cm]
& Unconstrained & $\begin{bmatrix} \bm{17450.5}&184&1 \end{bmatrix} \atop \begin{bmatrix} 3.3&0.1&19.5\\1.1&1&3 \end{bmatrix}$  & $\begin{bmatrix} \bm{3092.0}&126&1 \end{bmatrix} \atop \begin{bmatrix} 4.3&0.2&21.9\\1.7&1&5 \end{bmatrix}$ & $\begin{bmatrix} \bm{1099.5}&95&1 \end{bmatrix} \atop \begin{bmatrix} 5.9&0.7&24.7\\2.7&1&8 \end{bmatrix}$ & $\begin{bmatrix} \bm{354.5}&50&1 \end{bmatrix} \atop \begin{bmatrix} 8.6&1.8&26.2\\5.0&1&21 \end{bmatrix}$  \\[0.7cm]
& $\leq 12$ trades & $\begin{bmatrix} \bm{453.9}&12&3 \end{bmatrix} \atop \begin{bmatrix} 19.0&5.6&56.7\\13.8&3&55 \end{bmatrix}$ & $\begin{bmatrix} \bm{377.9}&12&3 \end{bmatrix} \atop \begin{bmatrix} 12.0&7.2&17.8\\22.5&11&45 \end{bmatrix}$ & $\begin{bmatrix} \bm{326.4}&12&2 \end{bmatrix} \atop \begin{bmatrix} 20.8&5.4&54\\18.4&3&49 \end{bmatrix}$ & $\begin{bmatrix} \bm{258.7}&12&2 \end{bmatrix} \atop \begin{bmatrix} 22.3&5.4&50.1\\19.0&3&46 \end{bmatrix}$  \\[0.7cm]
& $\geq 10$ days waiting & $\begin{bmatrix} \bm{440.8}&17&10 \end{bmatrix} \atop \begin{bmatrix} 12.9&3.8&33.5\\11.7&5&16 \end{bmatrix}$ & $\begin{bmatrix} \bm{341.2}&16&10 \end{bmatrix} \atop \begin{bmatrix} 11.5&2.1&31.4\\ 11.7&5&16 \end{bmatrix}$ & $\begin{bmatrix} \bm{269.3}&15&10 \end{bmatrix} \atop \begin{bmatrix} 11.1&1.8&29.5\\ 12.4&5&16 \end{bmatrix}$ & $\begin{bmatrix} \bm{186.6}&12&10 \end{bmatrix} \atop \begin{bmatrix} 17.9&12.1&26.7\\ 18.5&11&36 \end{bmatrix}$ \\[0.55cm]
\colrule
\multirow{4}{*}{\begin{turn}{90} \hspace{-2cm} \large{$q=1$} \end{turn}} & 
Buy-and-Hold & $\begin{bmatrix} \bm{3.3}&1&199 \end{bmatrix}$ & $\begin{bmatrix} \bm{2.8}&1&199 \end{bmatrix}$ & $\begin{bmatrix} \bm{2.3}&1&199 \end{bmatrix}$ & $\begin{bmatrix} \bm{1.3}&1&199 \end{bmatrix}$ \\[0.2cm]
& Unconstrained & $\begin{bmatrix} \bm{3671.6}&190&1 \end{bmatrix} \atop \begin{bmatrix} 2.4&0.02&7.5\\1.0&1&2 \end{bmatrix}$  & $\begin{bmatrix} \bm{769.0}&123&1 \end{bmatrix} \atop \begin{bmatrix} 3.3&0.4&9.0\\1.8&1&6 \end{bmatrix}$ & $\begin{bmatrix} \bm{378.5}&76&1 \end{bmatrix} \atop \begin{bmatrix} 4.8&1.2&9.0\\3.3&1&9 \end{bmatrix}$ & $\begin{bmatrix} \bm{173.2}&38&1 \end{bmatrix} \atop \begin{bmatrix} 8.0&2.8&22.5\\7.0&2&19 \end{bmatrix}$  \\[0.7cm]
& $\leq 12$ trades & $\begin{bmatrix} \bm{227.0}&12&1 \end{bmatrix} \atop \begin{bmatrix} 12.8&0.5&36.0\\15.5&1&43 \end{bmatrix}$ & $\begin{bmatrix} \bm{212.3}&12&1 \end{bmatrix} \atop \begin{bmatrix} 14.5&2.6&34.5\\16.2&1&43 \end{bmatrix}$ & $\begin{bmatrix} \bm{173.0}&12&4 \end{bmatrix} \atop \begin{bmatrix} 13.8&2.6&33.2\\18.3&7&43 \end{bmatrix}$ & $\begin{bmatrix} \bm{129.6}&12&3 \end{bmatrix} \atop \begin{bmatrix} 18.7&9.5&30.5\\23.7&7&43 \end{bmatrix}$  \\[0.7cm]
& $\geq 10$ days waiting & $\begin{bmatrix} \bm{267.5}&17&10 \end{bmatrix} \atop \begin{bmatrix} 10.4&2.0&18.4\\11.2&1&15 \end{bmatrix}$ & $\begin{bmatrix} \bm{200.5}&17&10 \end{bmatrix} \atop \begin{bmatrix} 9.8&3.8&17.2\\ 11.5&10&14 \end{bmatrix}$ & $\begin{bmatrix} \bm{153.2}&17&10 \end{bmatrix} \atop \begin{bmatrix} 9.2&1.5&14.5\\ 11.0&8&14 \end{bmatrix}$ & $\begin{bmatrix} \bm{113.3}&15&10 \end{bmatrix} \atop \begin{bmatrix} 9.1&1.2&22.1\\ 14.3&6&33 \end{bmatrix}$ \\[0.55cm]
\colrule
\multirow{4}{*}{\begin{turn}{90} \hspace{-2cm} \large{$q=2$} \end{turn}} & 
Buy-and-Hold & $\begin{bmatrix} \bm{1.7}&1&199 \end{bmatrix}$ & $\begin{bmatrix} \bm{0.6}&1&199 \end{bmatrix}$ & $\begin{bmatrix} \bm{-0.4}&1&199 \end{bmatrix}$ & $\begin{bmatrix} \bm{-2.4}&1&199 \end{bmatrix}$ \\[0.2cm]
& Unconstrained & $\begin{bmatrix} \bm{1882.1}&191&1 \end{bmatrix} \atop \begin{bmatrix} 2.0&2\text{e}-14&6.9\\1.0&1&3 \end{bmatrix}$  & $\begin{bmatrix} \bm{459.7}&113&1 \end{bmatrix} \atop \begin{bmatrix} 3.0&0.1&11.4\\2.0&1&8 \end{bmatrix}$ & $\begin{bmatrix} \bm{229.3}&66&1 \end{bmatrix} \atop \begin{bmatrix} 4.9&1.2&12.0\\3.7&1&10 \end{bmatrix}$ & $\begin{bmatrix} \bm{102.9}&37&1 \end{bmatrix} \atop \begin{bmatrix} 6.4&2.9&10.0\\6.7&1&16 \end{bmatrix}$  \\[0.7cm]
& $\leq 12$ trades & $\begin{bmatrix} \bm{186.3}&12&1 \end{bmatrix} \atop \begin{bmatrix} 12.9&7.2&31.8\\16.6&1&67 \end{bmatrix}$ & $\begin{bmatrix} \bm{158.0}&12&1 \end{bmatrix} \atop \begin{bmatrix} 11.8&5.8&30.6\\15.3&1&49 \end{bmatrix}$ & $\begin{bmatrix} \bm{137.9}&12&1 \end{bmatrix} \atop \begin{bmatrix} 11.9&7.5&24.6\\18.0&5&37 \end{bmatrix}$ & $\begin{bmatrix} \bm{89.0}&12&3 \end{bmatrix} \atop \begin{bmatrix} 10.7&5.9&18.3\\18.0&8&36 \end{bmatrix}$  \\[0.7cm]
& $\geq 10$ days waiting & $\begin{bmatrix} \bm{215.2}&18&10 \end{bmatrix} \atop \begin{bmatrix} 9.2&3.3&13.9\\10.8&9&14 \end{bmatrix}$ & $\begin{bmatrix} \bm{164.1}&18&10 \end{bmatrix} \atop \begin{bmatrix} 7.5&2.5&12.4\\ 10.6&8&14 \end{bmatrix}$ & $\begin{bmatrix} \bm{116.2}&15&10 \end{bmatrix} \atop \begin{bmatrix} 9.0&4.9&12.3\\ 13.3&10&20 \end{bmatrix}$ & $\begin{bmatrix} \bm{70.7}&11&10 \end{bmatrix} \atop \begin{bmatrix} 9.7&5.8&18.3\\ 17.6&10&29 \end{bmatrix}$ \\[0.55cm]
\colrule
\multirow{4}{*}{\begin{turn}{90} \hspace{-0.85cm} \large{Summary} \end{turn}} & 
Buy-and-Hold & $ \bm{23.2}$ & $\bm{20.4}$ & $\bm{17.8}$ & $\bm{12.5}$ \\[0.2cm]
& Unconstrained & $\bm{23004.2}$ & $\bm{4320.7}$ & $\bm{1707.3}$ & $\bm{630.6}$  \\[0.2cm]
& $\leq 12$ trades & $\bm{867.2}$ & $\bm{748.2}$ & $\bm{637.3}$ & $\bm{477.3}$  \\[0.2cm]
& $\geq 10$ days waiting & $\bm{923.5}$ & $\bm{705.8}$ & $\bm{538.7}$ & $\bm{370.6}$ \\[0.1cm]
\botrule
\end{tabular}}
\end{center}
\end{table}

\begin{table}
\begin{center}
\tbl{Summary of quantitative results of Example 2 for the second case: time-synchronized trading with diversification for all times.\label{tab3_ex2}}
{
\def\arraystretch{0.7}
\begin{tabular}{@{}llcccc}\toprule
& & $\epsilon=0$ & $\epsilon=0.5$ & $\epsilon=1$ & $\epsilon=2$ \\\colrule
\multirow{4}{*}{\begin{turn}{90} \hspace{-2cm} \large{$q=0$} \end{turn}} & 
Buy-and-Hold & $\begin{bmatrix} \bm{18.2}&1&199 \end{bmatrix}$ & $\begin{bmatrix} \bm{17.0}&1&199 \end{bmatrix}$ & $\begin{bmatrix} \bm{15.9}&1&199 \end{bmatrix}$ & $\begin{bmatrix} \bm{13.6}&1&199 \end{bmatrix}$ \\[0.2cm]
& Unconstrained & $\begin{bmatrix} \bm{17450.5}&184&1 \end{bmatrix} \atop \begin{bmatrix} 3.3&0.1&19.5\\1.1&1&3 \end{bmatrix}$  & $\begin{bmatrix} \bm{3092.0}&126&1 \end{bmatrix} \atop \begin{bmatrix} 4.3&0.2&21.9\\1.7&1&5 \end{bmatrix}$ & $\begin{bmatrix} \bm{1099.5}&95&1 \end{bmatrix} \atop \begin{bmatrix} 5.9&0.7&24.7\\2.7&1&8 \end{bmatrix}$ & $\begin{bmatrix} \bm{354.5}&50&1 \end{bmatrix} \atop \begin{bmatrix} 8.6&1.8&26.2\\5.0&1&21 \end{bmatrix}$  \\[0.7cm]
& $\leq 12$ trades & $\begin{bmatrix} \bm{454.3}&12&2 \end{bmatrix} \atop \begin{bmatrix} 19.0&7.0&56.7\\13.8&2&55 \end{bmatrix}$ & $\begin{bmatrix} \bm{377.9}&12&3 \end{bmatrix} \atop \begin{bmatrix} 19.0&5.4&54.3\\15.3&3&55 \end{bmatrix}$ & $\begin{bmatrix} \bm{326.4}&12&2 \end{bmatrix} \atop \begin{bmatrix} 20.8&5.4&54.0\\18.4&3&49 \end{bmatrix}$ & $\begin{bmatrix} \bm{258.7}&12&2 \end{bmatrix} \atop \begin{bmatrix} 22.3&5.4&50.1\\19.0&3&46 \end{bmatrix}$  \\[0.7cm]
& $\geq 10$ days waiting & $\begin{bmatrix} \bm{440.8}&17&10 \end{bmatrix} \atop \begin{bmatrix} 12.9&3.8&33.5\\11.7&5&16 \end{bmatrix}$ & $\begin{bmatrix} \bm{341.2}&16&10 \end{bmatrix} \atop \begin{bmatrix} 11.5&2.1&31.4\\ 11.7&5&16 \end{bmatrix}$ & $\begin{bmatrix} \bm{269.3}&15&10 \end{bmatrix} \atop \begin{bmatrix} 11.1&1.8&29.5\\ 12.4&5&16 \end{bmatrix}$ & $\begin{bmatrix} \bm{186.6}&12&10 \end{bmatrix} \atop \begin{bmatrix} 17.9&12.1&26.7\\ 18.5&11&36 \end{bmatrix}$ \\[0.55cm]
\colrule
\multirow{4}{*}{\begin{turn}{90} \hspace{-2cm} \large{$q=1$} \end{turn}} & 
Buy-and-Hold & $\begin{bmatrix} \bm{1.7}&1&199 \end{bmatrix}$ & $\begin{bmatrix} \bm{0.6}&1&199 \end{bmatrix}$ & $\begin{bmatrix} \bm{-0.4}&1&199 \end{bmatrix}$ & $\begin{bmatrix} \bm{-2.4}&1&199 \end{bmatrix}$ \\[0.2cm]
& Unconstrained & $\begin{bmatrix} \bm{3009.4}&177&1 \end{bmatrix} \atop \begin{bmatrix} 2.4&0.02&7.8\\1.1&1&3 \end{bmatrix}$  & $\begin{bmatrix} \bm{523.5}&89&1 \end{bmatrix} \atop \begin{bmatrix} 3.6&0.4&13.0\\2.5&1&8 \end{bmatrix}$ & $\begin{bmatrix} \bm{259.9}&46&1 \end{bmatrix} \atop \begin{bmatrix} 6.1&0.6&18.6\\5.8&1&17 \end{bmatrix}$ & $\begin{bmatrix} \bm{137.8}&24&1 \end{bmatrix} \atop \begin{bmatrix} 9.1&2.9&25.9\\10.3&2&23 \end{bmatrix}$  \\[0.7cm]
& $\leq 12$ trades & $\begin{bmatrix} \bm{142.8}&12&2 \end{bmatrix} \atop \begin{bmatrix} 8.9&5.0&18.0\\11.3&2&55 \end{bmatrix}$ & $\begin{bmatrix} \bm{112.7}&12&3 \end{bmatrix} \atop \begin{bmatrix} 8.2&0.8&16.8\\13.5&4&55 \end{bmatrix}$ & $\begin{bmatrix} \bm{87.7}&12&2 \end{bmatrix} \atop \begin{bmatrix} 9.8&5.4&19.7\\14.8&4&49 \end{bmatrix}$ & $\begin{bmatrix} \bm{67.3}&12&2 \end{bmatrix} \atop \begin{bmatrix} 11.1&3.9&22.4\\15.3&4&46 \end{bmatrix}$  \\[0.7cm]
& $\geq 10$ days waiting & $\begin{bmatrix} \bm{179.1}&17&10 \end{bmatrix} \atop \begin{bmatrix} 7.6&0.7&15.7\\11.7&5&16 \end{bmatrix}$ & $\begin{bmatrix} \bm{135.9}&14&10 \end{bmatrix} \atop \begin{bmatrix} 8.1&0.8&15.1\\ 12.8&5&24 \end{bmatrix}$ & $\begin{bmatrix} \bm{116.0}&12&10 \end{bmatrix} \atop \begin{bmatrix} 9.0&2.1&23.4\\ 16.6&10&41 \end{bmatrix}$ & $\begin{bmatrix} \bm{73.2}&9&10 \end{bmatrix} \atop \begin{bmatrix} 11.3&6.0&14.8\\ 19.8&11&32 \end{bmatrix}$ \\[0.55cm]
\colrule
\multirow{4}{*}{\begin{turn}{90} \hspace{-2cm} \large{$q=2$} \end{turn}} & 
Buy-and-Hold & $\begin{bmatrix} \bm{3.3}&1&250 \end{bmatrix}$ & $\begin{bmatrix} \bm{2.1}&1&250 \end{bmatrix}$ & $\begin{bmatrix} \bm{1.1}&1&250 \end{bmatrix}$ & $\begin{bmatrix} \bm{-0.9}&1&250 \end{bmatrix}$ \\[0.2cm]
& Unconstrained & $\begin{bmatrix} \bm{1640.8}&177&1 \end{bmatrix} \atop \begin{bmatrix} 2.1&0.02&10.6\\1.1&1&4 \end{bmatrix}$  & $\begin{bmatrix} \bm{342.9}&91&1 \end{bmatrix} \atop \begin{bmatrix} 3.2&0.1&9.5\\2.4&1&6 \end{bmatrix}$ & $\begin{bmatrix} \bm{170.3}&45&1 \end{bmatrix} \atop \begin{bmatrix} 5.8&0.4&11.4\\6.1&1&21 \end{bmatrix}$ & $\begin{bmatrix} \bm{87.9}&20&1 \end{bmatrix} \atop \begin{bmatrix} 7.9&3.9&11.7\\11.0&3&22 \end{bmatrix}$  \\[0.7cm]
& $\leq 12$ trades & $\begin{bmatrix} \bm{114.8}&12&2 \end{bmatrix} \atop \begin{bmatrix} 8.2&4.6&14.6\\11.6&2&55 \end{bmatrix}$ & $\begin{bmatrix} \bm{92.2}&11&3 \end{bmatrix} \atop \begin{bmatrix} 9.4&0.6&24.6\\15.7&3&59 \end{bmatrix}$ & $\begin{bmatrix} \bm{74.4}&11&2 \end{bmatrix} \atop \begin{bmatrix} 10.3&1.8&28.2\\18.3&3&53 \end{bmatrix}$ & $\begin{bmatrix} \bm{50.0}&7&6 \end{bmatrix} \atop \begin{bmatrix} 15.9&3.8&25.6\\36.0&6&53 \end{bmatrix}$  \\[0.7cm]
& $\geq 10$ days waiting & $\begin{bmatrix} \bm{134.2}&16&10 \end{bmatrix} \atop \begin{bmatrix} 6.4&0.3&11.7\\12.4&10&16 \end{bmatrix}$ & $\begin{bmatrix} \bm{102.9}&13&10 \end{bmatrix} \atop \begin{bmatrix} 8.4&0.5&12.2\\ 17.3&11&33 \end{bmatrix}$ & $\begin{bmatrix} \bm{87.8}&12&11 \end{bmatrix} \atop \begin{bmatrix} 7.3&4.5&10.1\\ 14.7&13&21 \end{bmatrix}$ & $\begin{bmatrix} \bm{52.8}&7&10 \end{bmatrix} \atop \begin{bmatrix} 13.9&8.1&19.3\\ 26.7&25&32 \end{bmatrix}$ \\[0.55cm]
\colrule
\multirow{4}{*}{\begin{turn}{90} \hspace{-0.85cm} \large{Summary} \end{turn}} & 
Buy-and-Hold & $ \bm{23.2}$ & $\bm{20.4}$ & $\bm{17.8}$ & $\bm{12.5}$ \\[0.2cm]
& Unconstrained & $ \bm{22100.7}$ & $\bm{3958.4}$ & $\bm{1529.7}$ & $\bm{580.2}$  \\[0.2cm]
& $\leq 12$ trades & $\bm{711.9}$ & $\bm{582.8}$ & $\bm{488.5}$ & $\bm{376.0}$  \\[0.2cm]
& $\geq 10$ days waiting & $ \bm{754.1}$ & $ \bm{580.0}$ & $\bm{473.1}$ & $\bm{312.6}$ \\[0.1cm]
\botrule
\end{tabular}}
\end{center}
\end{table}

We consider 16 currencies and 15 non-currency assets. Real-world data is obtained according to Table \ref{tab1_ex2}. We consider the time horizon August 5, 2015 until August 3, 2016. Because of different trading holidays in the different countries, a total of $199$ trading days could be determined common to all assets. We diversify in three assets at every trading time $t$, i.e., we set $Q=3$. 

We distinguish between two cases: \emph{synchronous} and \emph{asynchronous} trading. Quantitative results are summarized in Tables \ref{tab2_ex2} and \ref{tab3_ex2}, respectively. The results for all $Q$ trajectories are reported. The ``Summary''-section in Tables \ref{tab2_ex2} and \ref{tab3_ex2} reports the sum of returns of all $Q$ trajectories. Exemplary, results are further visualized in Figure \ref{fig:ex2fig3}.  The black-dashed horizontal line in the corresponding top subplots denotes $N_c=16$ to distinguish currency and non-currency asset investments. 

For performance comparison, we consider a Buy-and-Hold strategy, whereby an asset is bought initially and then held. The most performant non-currency assets from Table \ref{tab1_ex2} for the time frame of interest were, in order,  the IBOVESPA (BRA), the Dow Jones Russia GDR (RUS) and the S\&P 500 (USA). Associated returns are reported in Table \ref{tab2_ex2} and \ref{tab3_ex2}, where we attribute the IBOVESPA to $q=0$ and the other two assets to $q=1$ and $q=2$, respectively.  

Interpretation of results is in line with Section \ref{subsec_ex1}. In particular, the influence of transaction costs and the encouragement of frequent trading upon relatively small percentage gains is recurrent. 

A remark about computational complexity needs to be made. The total number of states, $N_z(N_t)$, without consideration of any heuristics is 6139, 71611 and 59905 for the three cases (unconstrained, constraint of at most $K=12$ trades, and constraint of waiting at least $D=10$ days between any two trades). These numbers can be computed according to the formulas stated in Section \ref{sec_multi_stage_opt_without_diversification}. Then, applying the heuristics from Section \ref{subsec_case2_NoDiversificationCstrt} and \ref{subsec_case3_NoDiversificationCstrt} to the given \texttt{finance.yahoo}-data trajectories, we measured (to give one example) $N_z(N_t)=65238$ and $N_z(N_t)=33161$ for the latter two cases, $q=0$ and $\epsilon=0$. Similar results are obtained for the other transaction cost levels and the other $q$-trajectories, resulting in overall computation times (for all $q=0,1,2$) in the tens of minutes. In contrast, for the unconstrained case, overall computation times for the generation of all $Q=3$ transition graphs were on average only slightly more than 10 seconds, thereby making the unconstrained case much more suitable for fast analysis of sets of multiple assets and foreign exchange rate trajectories. Secondly, the trajectories for $q=0$ are identical for both time-asynchronous and -synchronous trading. However, for  the remaining investment trajectories with $q>0$, the number of states is much lower for time-synchronous trading in comparison to the asynchronous case. For time-synchronous trading, $Q=3$ and a bound on the total admissable number of trades, the total number of states is 63803 for $q=0$, but 39713 and 23549 for $q=1$ and $q=2$, respectively. All numerical experiments throughout this paper were conducted on a laptop running Ubuntu 14.04 equipped with an Intel Core i7 CPU @ 2.80GHz$\times$8, 15.6 GB of memory, and using Python 2.7.

\begin{figure}
\centering
\includegraphics[width=0.9\textwidth]{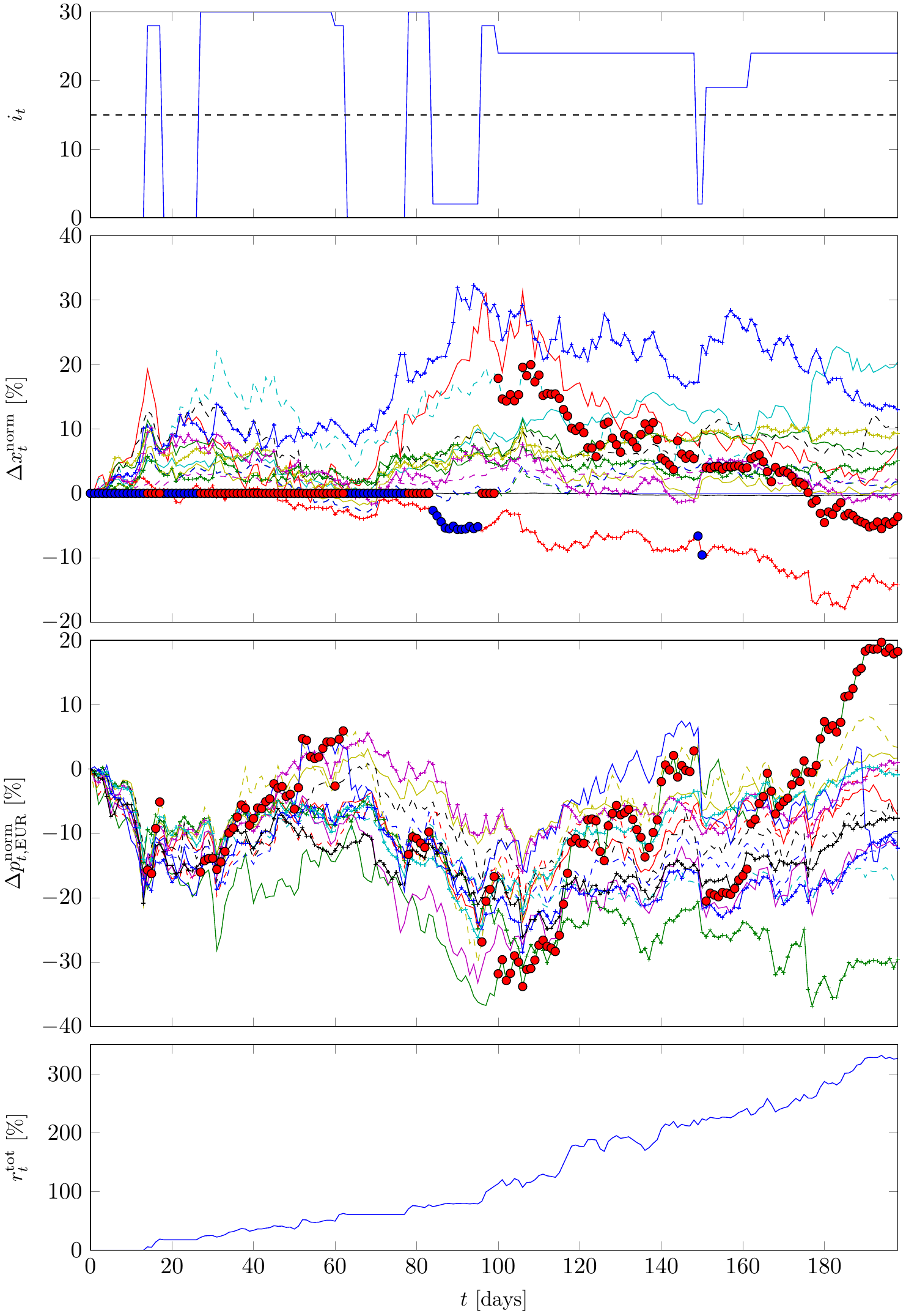}
\caption{Example 2. The results for $q=0$, at most $K=12$ admissable trades and transaction cost level $\epsilon=1$. For $q=0$, the results for asynchronous and synchronous trading are identical. See Section \ref{subsec_ex2} for interpretation.}
\label{fig:ex2fig3}
\end{figure}

\subsection{Example 3: A downtrending and an uptrending stock\label{subsec_ex3}}

\begin{table}
\begin{center}
\tbl{Summary of quantitative results of Example 3. Comparison of a downtrending and an uptrending stock for the time period between August 10, 2015 and August 8, 2016. The exemplary downtrending stock is of Deutsche Bank AG (\texttt{finance.yahoo}-symbol: \texttt{DBK.DE}). The exemplary uptrending stock is of Adidas AG (\texttt{finance.yahoo}-symbol: \texttt{ADS.DE}). See Section \ref{subsec_ex3} for interpretation.\label{tab_ex3}}
{
\def\arraystretch{0.7}
\begin{tabular}{@{}lc}\toprule
\texttt{DKB.DE} & $\epsilon=1$ \\\colrule
Buy-and-Hold & $\begin{bmatrix} \bm{-61.4}&1&260 \end{bmatrix}$ \\[0.2cm]
Unconstrained & $\begin{bmatrix} \bm{382.8}&51&1 \end{bmatrix} \atop \begin{bmatrix} 7.5&1.5&21.8\\4.2&1&11 \end{bmatrix}$ \\[0.55cm]
\botrule
\end{tabular}
\hspace{2cm} \begin{tabular}{@{}lc}\toprule
\texttt{ADS.DE} & $\epsilon=1$ \\\colrule
Buy-and-Hold & $\begin{bmatrix} \bm{95.6}&1&260 \end{bmatrix}$ \\[0.2cm]
Unconstrained & $\begin{bmatrix} \bm{249.4}&43&1 \end{bmatrix} \atop \begin{bmatrix} 7.0&1.0&20.1\\7.8&2&24 \end{bmatrix}$ \\[0.55cm]
\botrule
\end{tabular}
}
\end{center}
\end{table}

The ultimate example compares achievable performances for an examplary downtrending and an an uptrending stock. The exemplary downtrending stock is of Deutsche Bank AG (\texttt{finance.yahoo}-symbol: \texttt{DKB.DE}). The exemplary uptrending stock is of Adidas AG (\texttt{finance.yahoo}-symbol: \texttt{ADS.DE}). Both stocks are listed in the German stock index (DAX). The time-frame considered is August 10, 2015 until August 8, 2016. There are 260 potential trading days. Both stocks are traded in currency EUR. We thus find optimal investment trajectories when i) trading \texttt{DKB.DE} and EUR, and ii) trading \texttt{ADS.DE}  and EUR. We assume propotional costs of $1\%$ identical for buying and selling, i.e., $\epsilon = \epsilon_\text{buy}=\epsilon_\text{sell}$. We set $\beta=50$.
Results are summarized in Table \ref{tab_ex3} and Figure \ref{fig:ex3_figDBK}. 

Unexpectedly and remarkably, the yearly return associated with the optimal investment trajectory of the downtrending stock is higher than its uptrending counterpart: $382.8\%$ vs. $249.2\%$. Importantly, note that the corresponding Buy-and-Hold returns are $-61.4\%$ and $95.6\%$, respectively. While overall downtrending, the price of \texttt{DKB.DE} indicates temporary steep price increases. Furthermore, these occur mostly towards the second half of the time-period of interest, and  thus imply stronger return growth due to the already compounded portfolio wealth that is available for investing at that time (instead of the initial $m_0$). Naturally, without a posteriori knowledge of price evolutions, an uptrending stock such as \texttt{ADS.DE} offers the advantage that missing the right selling dates is less important. Interestingly, both downtrending and uptrending are traded optimally upon similar short-term average price increases: $7.5\%$ and $7\%$. Similarly, the optimal holding periods of the stocks are short with on average $4.2$ and $7.8$ days, respectively.

\begin{figure}
\centering
\input{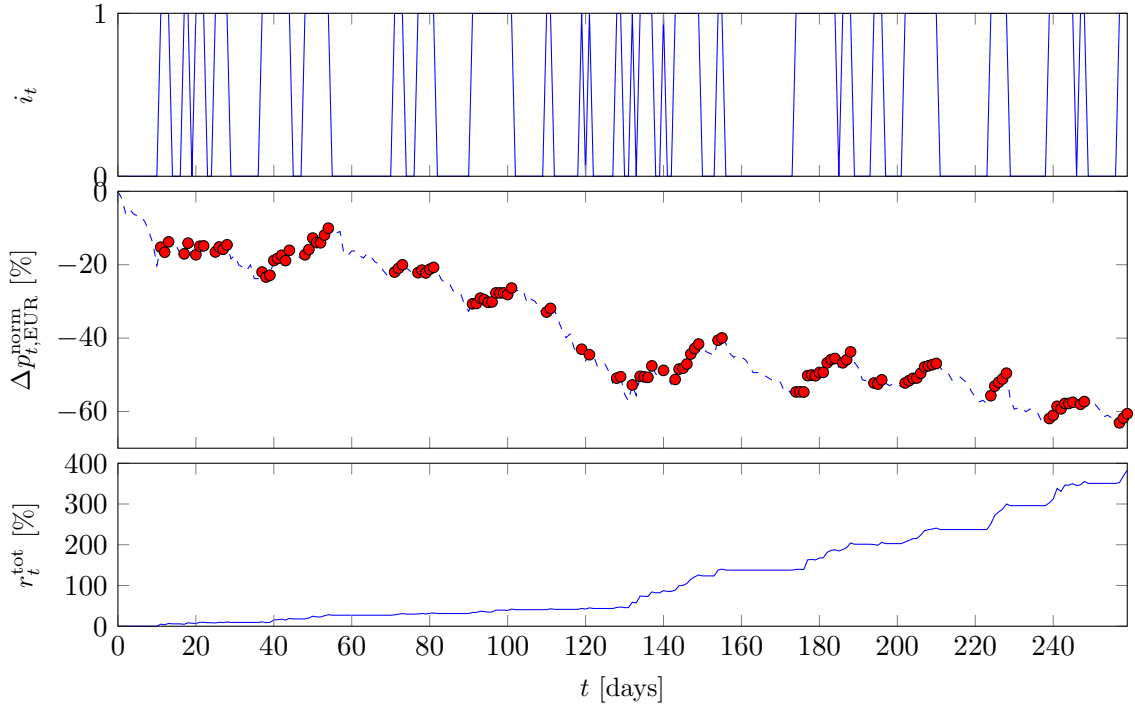}
\caption{The downtrending stock of Example 3 for the unconstrained trading case in case of proportional transaction costs of $1\%$, see Section \ref{subsec_ex3}. The exemplary downtrending stock is of Deutsche Bank AG (\texttt{finance.yahoo}-symbol: \texttt{DKB.DE}). See also Table \ref{tab_ex3}.}
\label{fig:ex3_figDBK}
\end{figure}


\section{Conclusion\label{sec_Conclusion}}

We developed a simple graph-based method for a posteriori (historical) multi-variate multi-stage optimal trading under transaction costs and a diversification constraint. Three variants were discussed, including  unconstrained trading frequency, a fixed number of total admissable trades, and the waiting of a specific time-period after every executed trade until the next trade. Findings were evaluated quantitatively on real-world data.

It was illustrated that transaction cost levels are decisive for achievable performance and significantly influence optimal trading frequency. Quantitative results further indicated optimal trading upon occasion rather than on fixed trading intervals, and, dependent on transaction cost levels, upon single- to low double-digit percentage gains with respect to the reference currency, and exploiting short-term trends. Achievable returns for optimized trading are uncomparably outperforming Buy-and-Hold strategies. Naturally, these returns are very difficult to achieve in practice without knowledge of future price and foreign exchange rate evolutions. 

The fundamental motivation and possibly best application of this work is to use it for i) the preparatory and \emph{automated labeling} of financial time-series data, which is almost unlimitedly available, and where transaction cost level $\epsilon$ can then be regarded as a hyperparameter for the desired tuning of labeled data, before ii) developing supervised machine learning applications for algorithmic trading and screening systems. This is subject of ongoing work.

\bibliography{rQUFguide}

\end{document}